\documentclass{article}
\usepackage{fullpage}

\usepackage{times}
\usepackage[utf8]{inputenc}

\usepackage{amsmath}
\usepackage{amsthm}
\usepackage{amssymb}
\usepackage{nicefrac}

\usepackage[colorlinks=true, allcolors=blue]{hyperref}

\usepackage{graphicx}
\usepackage{caption}
\usepackage{tikz}
\usetikzlibrary{arrows.meta,positioning,shadows.blur,calc,fit,backgrounds}

\usepackage{booktabs}

\usepackage[ruled,linesnumbered]{algorithm2e}
\usepackage[nameinlink]{cleveref}
\crefname{AlgoLine}{line}{lines}
\Crefname{AlgoLine}{Line}{Lines}

\usepackage{soul}
\usepackage{url}
\usepackage[dvipsnames]{xcolor}

\tikzset{
treenode/.style = {circle, draw=black, inner sep=2pt, minimum size=6mm, font=\small},
snode/.style = {rectangle, draw=black, rounded corners, inner sep=4pt, font=\small},
el/.style = {edge from parent/.style={draw, -latex}}
}

\usepackage{xfrac}
\usepackage{booktabs,multirow}
\usepackage{amsthm,amsmath,amssymb,amsfonts}
\usepackage{mathtools,thmtools}
\usepackage{thm-restate}
\usepackage{cleveref}
\usepackage[dvipsnames]{xcolor}
\usepackage{comment}
\usepackage{todonotes}
\usepackage{natbib}
\usepackage{tikz}
\usetikzlibrary{arrows.meta}
\usetikzlibrary{positioning}
\usetikzlibrary{calc}
\usepackage{xspace}
\usepackage{nicefrac}
\usepackage[shortlabels]{enumitem}
\usepackage{pdfpages}

\newtheorem{lemma}{Lemma}
\theoremstyle{definition}
\newtheorem{definition}{Definition}

\newtheorem{open}{Open Question}
\usepackage{tcolorbox}

\makeatletter
\newcommand{\leqnomode}{\tagsleft@true}
\newcommand{\reqnomode}{\tagsleft@false}
\makeatother

\renewcommand{\paragraph}[1]{\smallskip\noindent\textbf{#1}}

\renewcommand{\le}{\leqslant}
\renewcommand{\leq}{\leqslant}
\renewcommand{\ge}{\geqslant}
\renewcommand{\geq}{\geqslant}

\newcommand{\bbR}{\mathbb{R}}
\newcommand{\bbN}{\mathbb{N}}

\DeclarePairedDelimiter\ceil{\lceil}{\rceil}
\DeclarePairedDelimiter\set{\{}{\}}           

\renewcommand{\vec}[1]{#1}
\newcommand{\vals}{\vec{v}}

\newcommand{\calI}{\mathcal{I}} 
\newcommand{\calM}{\mathcal{M}} 
\newcommand{\calL}{\mathcal{L}} 
\newcommand{\calC}{\mathcal{C}} 
\renewcommand{\top}{\operatorname{top}}

\renewcommand{\cite}{\citep}

\title{Constrained Fair Allocations via Partition Matroid Reductions}
\author{Benjamin Cookson \and Nisarg Shah}
\date{University of Toronto\\\{bcookson, nisarg\}@cs.toronto.edu}

\begin{document}

\maketitle

\begin{abstract}
We study fair allocation of indivisible goods under additive valuations and matroid constraints. A challenging open question is whether a complete and feasible envy-free up to one good (EF1) allocation exists under every matroid that admits a complete and feasible allocation. The state-of-the-art result by \citet{biswas2018fair} positively resolves this question for partition matroids.

Our first result positively resolves it for laminar matroids, which generalize partition matroids, when there are three agents. Our technique reduces this general existence question to finding an EF1 allocation satisfying a mild additional condition under a single finite-sized key laminar matroid, and we establish the required allocation by case analysis. We show that our technique somewhat extends to four agents, reducing the analogous problem to finding EF1 allocations under two finite-sized laminar matroids, although we are unable to establish their existence. We also use recent matroid decomposition results to establish EF1 existence under broader classes of matroids. Specifically, we show that EF1 allocations always exist under transversal matroids whenever a complete allocation is feasible, and obtain existence results for graphic matroids and gammoids under stronger assumptions.
\end{abstract}

\section{Introduction}

Fair allocation of indivisible goods has long been a central problem in economics and computer science. In its most basic form, it asks how to partition a set of $m$ indivisible goods among $n$ agents with additive valuations, i.e., when their value for a set of goods is their total value for the individual goods in that set.

While this basic model offers significant algorithmic insights on fairness, it suffers from limited direct applicability as many real-world settings impose additional structure on the desired allocation. For example, a museum may wish to allocate exhibit items between its several branches, but must additionally ensure that no branch receives too many items from any single category (e.g., paintings or sculptures). Mathematically, this is depicted as a \emph{partition matroid constraint}, where the bundle allocated to each branch (agent) must be an independent set of a partition matroid over the exhibit items (goods) induced by the categories and their upper limits. Some constraints are more complex. When a hospital assigns shifts to nurses, it may wish to not only limit the number of shifts on any given day of the week, but also the total number of shifts allocated in the entire week. Similarly, when a conference matches reviewers to submissions, it may model having multiple reviewers per submission by creating copies of the submission but enforcing that each reviewer gets at most one copy of each submission, in addition to limiting the total review load. These examples can be viewed as \emph{laminar matroid constraints}, which generalize partition matroids by imposing upper bounds with respect to a laminar family of sets in which any two sets are either disjoint or one is contained in the other. Assigning courses to students at a university can involve additional layers as students face total course load limits, departmental caps, and sometimes fine-grained constraints on courses from different sub-fields, yielding even richer laminar structures.

A substantial literature studies fair allocation under such constraints. A prominent open question is a strikingly simple one: \emph{Does a complete and feasible envy-free up to one good (EF1) allocation always exist subject to matroid constraints that admit a complete and feasible allocation?} EF1 is arguably the most compelling fairness criterion for allocating indivisible goods~\cite{LMMS04,Bud11,CKMP+19}, which demands that no agent strictly prefer the bundle assigned to another agent with her favorite good removed to her own bundle. This criterion is surprisingly flexible: it can be attained in conjunction with Pareto optimality~\cite{CKMP+19} and $\nicefrac{2}{3}$-maximin share fairness~\cite{akrami2025achieving}, allows achieving exact envy-freeness in expectation by randomizing over EF1 allocations~\cite{AFSV23}, and continues to exist under monotone valuations~\cite{LMMS04} or when allocating (undesirable) chores together with goods~\cite{aziz2022fair}. Therefore, it is striking that its existence remains a major unresolved question in the literature on constrained fair division.

Progress on this question has so far been limited, with the state-of-the-art result by \citet{biswas2018fair} proving that EF1 allocations exist under partition matroid constraints and additive valuations. See \Cref{sec:related} for an overview of further related work. In this paper, we make progress on this fundamental problem by proving several new existence results for laminar (and broader classes of) matroid constraints using a technique called \emph{partition matroid reductions}.

\subsection{Our Contributions}\label{sec:results}

In this work, we establish new EF1 guarantees for several classes of matroid constraints. These results all revolve around a common technique, in which we take a complex matroid and reduce it to a simple matroid where the desired guarantee is already known (such as a partition matroid) or is quite simple to prove.

Our first result, presented in \Cref{sec:section3}, positively resolves this question for three agents under laminar matroid constraints. Specifically, we show that whenever a laminar matroid admits a complete and feasible allocation to three agents, it also admits a complete and feasible EF1 allocation to three agents. We demonstrate that finding such an allocation under any laminar matroid reduces to establishing a slightly stronger condition than EF1 under a specific, finite-sized \emph{key laminar matroid}. We are able to prove such a condition for this key matroid, and we show how this result combines with the algorithm of \citet{biswas2018fair} to yield a polynomial-time solution for this problem.

Additionally, in \Cref{app:four-agents}, we investigate how our techniques extend to instances with laminar matroid constraints and more than $3$ agents. We show that our decomposition technique extends to the four-agent setting, effectively reducing the general problem to finding EF1 allocations, together with an additional condition on envy between agents, under two specific finite-sized laminar matroids. However, resolving these base cases remains an open challenge. In fact, we show that the current techniques we used for solving the key matroids in the three-agent case are insufficient for $n=4$.

Finally, we examine fair allocation under several other classes of matroid constraints. We leverage recent matroid theory results from \citet{berczi2021list}, which imply that certain classes of matroids can always be reduced to partition matroids in such a way that finding an EF1 allocation for the partition matroid yields one for the original matroid. Combining this reduction with the partition-matroid algorithm of \citet{biswas2018fair}, we obtain EF1 allocations for $n$ agents under the following assumptions:

\begin{itemize}
    \item For transversal matroids: the matroid's ground set can be partitioned into $n$ independent sets.
    \item For graphic matroids: the matroid's ground set can be partitioned into $\lfloor(n+1)/2\rfloor$ independent sets.
    \item For gammoids: the matroid's ground set can be partitioned into $\lfloor(n+2)/2\rfloor$ independent sets.
\end{itemize}

For transversal matroids, this result fully resolves the EF1 matroid conjecture for this class: if a complete allocation to $n$ agents is feasible, then an EF1 allocation to those same $n$ agents also exists.

For graphic matroids and gammoids, the result gives a weaker version of the conjecture: for $n$ agents, the matroid must satisfy the stronger colorability assumption stated above rather than merely being $n$-colorable.

\subsection{Related Work}\label{sec:related}

As stated previously, \citet{biswas2018fair} initiated the search for EF1 under matroid constraints. In their original paper, they presented an elegant, polynomial-time algorithm that finds an EF1 allocation subject to partition matroid constraints whenever a complete allocation is feasible. In the same paper, Biswas and Barman showed that, when agents have identical valuations, the same existence guarantee holds for laminar matroid constraints. Later, in~\cite{biswas2019matroid}, the same authors extended the latter result to all base-orderable matroid constraints, a much broader class than laminar matroids.

Going beyond base-orderable matroids when agents have identical valuations, \citet{akrami2025matroids} resolve the matroid equitability conjecture and, as a consequence, show that an EF1 allocation exists when agents have identical tri-valued valuations and a complete allocation is feasible. Further, they show that if Gabow's Conjecture \cite{gabow1976decomposing} (a long-standing conjecture in matroid theory) is true, then the same EF1 guarantee holds under identical valuations subject to any matroid constraint under which a complete allocation is feasible.

When agents have non-identical valuation functions, \citet{cookson2025constrained} show that the allocation that maximizes Nash Welfare achieves a $\nicefrac{1}{2}$ approximation of EF1 while also achieving the efficiency notion of Pareto Optimality. \citet{shoshan2022efficient} also study matroid-constrained allocations that are simultaneously fair and efficient, showing that an EF1 and PO allocation always exists for two agents under any partition matroid constraint under which a complete allocation is feasible.

\citet{dror2023fair} also study EF1 under matroid constraints. Their main results are regarding heterogeneous matroid constraints, where each agent's bundle is constrained by a different matroid. When there is a single matroid constraint for all agents, they show that, when agents have (non-identical) binary valuations, an EF1 allocation always exists for three agents subject to any base-orderable matroid constraint under which a complete allocation is feasible. Particularly relevant to our work, \citet{dror2023fair} also point out that, for any class of matroids in which the desired EF1 guarantee is known under identical preference functions whenever a complete allocation is feasible, the corresponding result for two agents with different preferences follows from a simple cut-and-choose algorithm. This establishes the guarantee for two agents subject to any base-orderable matroid (and thus any laminar matroid) under which a complete allocation is feasible.

Finally, many other classes of constraints other than matroid constraints have been studied in the context of fair division. For a comprehensive overview of this literature, we direct the reader to \cite{suksompong2021constraints}.

\paragraph{Independent Work:} Independently of this work, \citet{equb26} also proved the existence of an EF1 allocation for $3$ agents under any laminar matroid constraint under which a complete allocation is feasible. Our techniques for proving this result differ slightly, but both hinge on identifying the same structural bottlenecks that make laminar matroids hard to solve. Both our algorithms highlight that the problem can essentially be reduced to finding an allocation with certain properties for a very specific laminar-constrained instance with $6$ goods (named the \emph{key matroid} in this paper, and discussed in detail in \Cref{sec:key}). Both this paper and the paper of \citet{equb26} provide essentially the same proof showing that the required allocation over these key matroid instances always exists. Then, with a subroutine for solving key matroid instances in hand, both our algorithm and that of \citet{equb26} rely on modifications of the partition matroid allocation algorithm of \citet{biswas2018fair} that use this subroutine to produce the final fair allocation.

Where the two results differ is in their details and presentation. The algorithm of \citet{equb26} takes a bottom-up approach, starting at the leaves of the laminar family and solving these leaves in a similar fashion to how \citet{biswas2018fair} solve individual partitions in their algorithm. Then, observing that after solving these partitions, any leftover unallocated items will need to be allocated subject to key matroid constraints in order to maintain feasibility, it completes its allocation using the designed subroutine.

In contrast, our algorithm takes a top-down approach. We first systematically decompose the input laminar matroid constraint into a new constraint that is a collection of disjoint uniform matroids and key matroids. We solve this reduced matroid by applying the algorithm of \citet{biswas2018fair} directly on the uniform matroids, and use the new subroutine to resolve the key matroids.

Ultimately, our goal with our laminar matroid algorithm and the subsequent results in \Cref{sec:partition-reductions} is to highlight the usefulness of reducing complex matroid constraints to simple matroid constraints, then finding EF1 allocations subject to the simple constraints. With this in mind, we believe our top-down approach to the laminar problem highlights this, and helps give context to exactly what makes laminar matroid constraints harder than partition matroid constraints. We hope that this perspective can be used to resolve laminar matroids for more than $3$ agents, and to generalize to larger classes of matroids than just the ones discussed in this paper.

\section{Preliminaries}\label{sec:prelims}

For any $r \in \mathbb{N}$, define $[r] := \set{1,2,\ldots,r}$. Let $N = [n]$ be a set of \emph{agents}, and $M$ be a set of $m$ \emph{(indivisible) goods}. Each agent $i$ has a valuation function $v_i : M \to \bbR_{\ge 0}$, where $v_i(g)$ is her value for good $g$. We assume additive valuations: with slight abuse of notation, the value of agent $i$ for a set of goods $S \subseteq M$ is $v_i(S) := \sum_{g \in S} v_i(g)$. Let $\vals = (v_1,\ldots,v_n)$ be the \emph{valuation profile}. In an \emph{allocation} $A = (A_1,\ldots,A_n)$, $A_i \subseteq M$ is the bundle of goods assigned to agent $i$ and $A_i \cap A_j = \emptyset$ for all distinct $i,j \in N$. We say that $A$ is \emph{complete} if $\cup_{i \in N} A_i = M$. The utility to agent $i$ under this allocation is $v_i(A_i)$.

\subsection{Matroid Constraints}

We are interested in \emph{matroid-constrained} fair division, in which we are allowed to assign a bundle to an agent only if the bundle is an independent set of a provided matroid over the set of goods.

Formally, a matroid over the set of goods $M$ is defined as $(M,\calI)$, where $\calI \subseteq 2^M$ satisfies: (a) $\emptyset \in \calI$; (b) \emph{(hereditary property)} if $S \in \calI$ and $S' \subseteq S$, then $S' \in \calI$; and (c) \emph{(exchange property)} if $S,T \in \calI$ and $|T| > |S|$, then there exists $g \in T \setminus S$ such that $S \cup \set{g} \in \calI$.

A matroid $(M,\calI)$ is called $k$-colorable if there exists a partition $(S_1,\ldots,S_k)$ of $M$ such that $S_i \in \calI$ for all $i \in [k]$. An allocation $A$ is called \emph{feasible} if $A_i \in \calI$ for all agents $i \in N$. Hence, in a matroid-constrained fair division instance, a complete and feasible allocation to $n$ agents exists if and only if the matroid is $n$-colorable. Hereinafter, when we refer to an allocation, we mean a complete and feasible allocation unless stated otherwise.

We will refer to the following popular families of matroids in this paper.

\begin{itemize}
    \item \textbf{Uniform matroid:} Given the ground set $M$ and an integer capacity $c \in \bbN$, set $\calI = \set{S \subseteq M : |S| \le c}$. This is one of the simplest matroids.

    \item \textbf{Partition matroid (cardinality constraints):} Let $\mathcal{M} = (M_1,\ldots,M_k)$ be a partition of $M$ into disjoint sets (``categories'') and let $u : \calM \to \bbN$ provide their corresponding capacities. Then, set $\calI = \set{S \subseteq M : |S \cap M_j| \le u(M_j) \text{ for all } j \in [k]}$. Hence, a partition matroid is the direct sum of uniform matroids on its parts.

    \item \textbf{Laminar matroid:} Let $\calL$ be a \emph{laminar family} of subsets of $M$ (i.e., for all $X, Y \in \calL$, $X \subseteq Y$, $Y \subseteq X$, or $X \cap Y = \emptyset$) and $u : \calL \to \bbN$ provide their corresponding capacities. Then, set $\calI = \set{S \subseteq M : |S \cap X| \le u(X) \text{ for all } X \in \calL}$.

    \item \textbf{Transversal matroid:} Let $G$ be an undirected bipartite graph with parts $M$ and $U$. Then, set $\calI$ to be the set of all subsets of $M$ that can be covered by a matching of $G$. Formally, given a set $U$ and its subsets $B_1, \ldots, B_m$, let $\calI$ be the collection of sets $S \subseteq M$ for which there exists an injective function $f: S \to U$ with $f(i) \in B_i$ for all $i \in S$.

    \item \textbf{Gammoid:} Let $D = (V, E)$ be a directed graph and $S, T \subseteq V$ be disjoint sets of source and terminal vertices, respectively. Then, let $\calI$ be the set of all sets $I \subseteq T$ such that there exist $|I|$ vertex-disjoint paths connecting $S$ to $I$.

    \item \textbf{Graphic Matroid:} Let $G = (V, M)$ be an undirected graph (note that the ground set $M$ is now the set of edges). Then, let $\calI$ be the set of all sets $S \subseteq M$ that do not contain a cycle (i.e., are forests).
\end{itemize}

These classes are related as follows:
\begin{itemize}
    \item Uniform $\subset$ Partition $\subset$ Laminar $\subset$ Gammoid;
    \item Uniform $\subset$ Transversal $\subset$ Gammoid.
\end{itemize}

As our main result concerns laminar matroids, let us introduce additional notation for them. Consider the laminar matroid induced by the ground set $M$, a laminar set family $\calL$ consisting of subsets of $M$, and a function $u : \calL \to \bbN$ yielding capacities. We denote it as $(M,\calL,u)$ and write its family of independent sets as $\calI(M,\calL,u)$, dropping the arguments in parentheses when the matroid is clear from the context. For each $S \in \calL$, define two quantities.
\begin{itemize}
    \item Let $D(S) = \set{S' \in \calL : S' \subsetneq S}$. We will refer to these as the \emph{descendants} of $S$.
    \item Let $C(S)$ be the set of inclusion-wise maximal elements of $D(S)$, i.e., the set of those $S' \in D(S)$ for which there is no $S'' \in D(S)$ with $S' \subsetneq S''$. We will refer to these as the \emph{children} of $S$.
\end{itemize}

The reasoning behind this nomenclature will become clear in \Cref{sec:packing}.

Finally, for all these constrained fair allocation problems, we define \emph{fairness} as envy-freeness up to one good (EF1).

\begin{definition}[Envy-freeness up to one good (EF1)]
    An allocation $A$ is \emph{envy-free up to one good} if, for every pair of agents $i,j \in N$, either $A_j = \emptyset$ or $v_i(A_i) \ge v_i(A_j \setminus \set{g})$ for some $g \in A_j$. In words, agent $i$ should not envy agent $j$, after excluding some good from agent $j$'s bundle.
\end{definition}

\section{Laminar Matroids With 3 Agents}\label{sec:section3}

In this section, we prove our main result.

\begin{restatable}{theorem}{main}\label{thm:n=3}
    For $n=3$ agents with additive valuations and laminar matroid constraints, a complete and feasible EF1 allocation exists whenever a complete and feasible allocation exists (i.e., whenever the matroid is $3$-colorable), and such an allocation can be computed in polynomial time.
\end{restatable}

First, we introduce a technique to reduce the given matroid to a simpler matroid subject to which a fair allocation can be found.
\begin{definition}[Feasibility-Preserving Operation]
    We say that an operation which takes an $n$-colorable laminar matroid $(M,\calL,u)$ as input and outputs $(M,\calL',u')$ is \emph{feasibility-preserving} if:
    \begin{enumerate}
        \item (Validity) $(M,\calL',u')$ is a laminar matroid, i.e., $\calL'$ is a laminar set family;
        \item (Rank Preservation) $(M,\calL',u')$ is $n$-colorable, i.e., the new matroid also admits a complete and feasible allocation; and
        \item (Independent Set Shrinkage) $\mathcal{I}(M, \calL',u') \subseteq \mathcal{I}(M, \calL,u)$, i.e., every independent set of the new matroid is also an independent set of the original matroid.
    \end{enumerate}
\end{definition}

The Rank Preservation condition ensures that all goods can still be allocated feasibly: because the resulting matroid remains $n$-colorable, it admits a complete and feasible allocation to the $n$ agents. The Independent Set Shrinkage condition instead concerns bundle feasibility; it ensures that any allocation that is feasible subject to the resulting matroid remains feasible subject to the original matroid. In our proof, we perform a sequence of feasibility-preserving operations on the given laminar matroid and then find an EF1 allocation subject to the resulting matroid. Completeness is unaffected because these operations retain the same ground set, and Independent Set Shrinkage guarantees that the allocation remains feasible subject to the original matroid.

In more detail, our proof for \Cref{thm:n=3} consists of four stages.

In \textbf{Stage 1}, we introduce a feasibility-preserving operation to completely \emph{pack} the laminar matroid---specifically, this ensures that for each $S \in \calL$ that has descendants, each good in $S$ appears in some descendant of $S$.

In \textbf{Stage 2}, we introduce a set of feasibility-preserving operations which are applied repeatedly to \emph{simplify} the packed matroid from Stage 1. These operations are repeated until no such operation can be applied.

In \textbf{Stage 3}, we show that the packed and simplified matroid resulting from Stage 2 consists of connected components that are either uniform matroids or a special laminar matroid over $6$ goods, which we term \emph{key matroid}. Using a result of \citet{biswas2018fair}, the general EF1 problem for laminar matroids reduces to finding an EF1 allocation with an additional ordered no-envy condition subject to just this single key matroid.

In \textbf{Stage 4}, we ``solve'' the key matroid, proving that it always admits the desired allocation regardless of the agent valuations.

Our method is constructive and results in a polynomial-time algorithm to find an EF1 allocation. Let us now understand all four stages in detail.

\subsection{Stage 1: Packing}\label{sec:packing}

\begin{definition}[Packed Laminar Set Family]
We say that a laminar set family $\calL$ over a ground set $M$ is \textbf{packed} if,
for every $S \in \calL$, either $D(S) = \emptyset$ or $\bigcup D(S) = S$.
\end{definition}

The following is a more useful characterization of the packed condition.

\begin{restatable}{lemma}{PackedChildrenPartition}\label{lem:packed-children-partition}
    A laminar set family $\calL$ is packed if and only if, for each $S \in \calL$, either $C(S) = \emptyset$ or $C(S)$ forms a disjoint partition of $S$.
\end{restatable}

\begin{proof}
    The ``if'' direction is trivial. If $C(S) = \emptyset$, then $D(S) = \emptyset$, and if $C(S)$ forms a disjoint partition of $S$, then $\bigcup D(S) \supseteq \bigcup C(S) = S$, which implies $\bigcup D(S) = S$.

    For the ``only if'' direction, take any $S \in \calL$. If $C(S) = \emptyset$, we are done. Suppose $C(S) \neq \emptyset$. First, we observe that $\bigcup C(S) = S$. Pick any $g \in S$. Since $\bigcup D(S) = S$, there exists $S' \in D(S)$ such that $g \in S'$. Pick any inclusion-wise maximal (e.g., maximum-cardinality) set in $D(S)$ that contains $g$. Then, $S' \in C(S)$. Hence, $g \in \bigcup C(S)$. Since this holds for any $g \in S$, we have $\bigcup C(S) = S$. Next, we observe that for any $S',T' \in C(S)$, $S' \cap T' = \emptyset$. The laminar structure implies that $S' \cap T' = \emptyset$, $S' \subset T'$, or $T' \subset S'$. Then, maximality of every child in $C(S)$ implies that it must be $S' \cap T' = \emptyset$.
\end{proof}

The structure in \Cref{lem:packed-children-partition} allows one to view a packed laminar set family $\calL$ as a forest, with each connected component being a rooted tree. Specifically, the sets in $\calL$ are the vertices, and for each $S \in \calL$, $C(S)$ are the \emph{children} of $S$ and $D(S)$ are the \textbf{descendants} of $S$ (i.e., leaves are precisely those $S \in \calL$ with $D(S) = \emptyset$).

Finally, we can introduce packed laminar matroids.
\begin{definition}[Packed Laminar Matroid]
We say that a laminar matroid $(M,\calL,u)$ is $n$-\emph{packed} if:
\begin{itemize}
    \item $\calL$ is packed, and
    \item $u(S) = \ceil{\nicefrac{|S|}{n}}$ for all $S \in \calL$.
\end{itemize}
\end{definition}

The intuition behind the value of $u(S)$ is as follows. It is known that $n$-colorability demands $u(S) \ge \ceil{\nicefrac{|S|}{n}}$ for all $S \in \calL$. So, if we impose the more stringent equality condition, a feasible allocation subject to that would remain feasible subject to the original matroid.

Let us now introduce the feasibility-preserving operation which can pack any laminar matroid.
\begin{definition}[Laminar Packing]
Given a positive integer $n$ and an $n$-colorable laminar matroid $(M,\calL,u)$, construct a laminar matroid $(M,\calL',u')$ as follows.
\begin{itemize}
    \item Let $\calL_0 = \calL \cup \set{M}$ and, for each $S \in \calL_0$, let $D_0(S) = \set{T \in \calL_0 : T \subsetneq S}$.
    \item Let $\calL' = \calL_0$. For each $S \in \calL_0$, if $0 < \big|\bigcup D_0(S)\big| < |S|$, then add $S^* = S \setminus \bigcup D_0(S)$ to $\calL'$.
    \item Set $u'(S) = \ceil{\nicefrac{|S|}{n}}$ for all $S \in \calL'$.
\end{itemize}
\end{definition}

The operation first places the entire ground set $M$ in the laminar family (this is done to ensure that every good from $M$ appears in the packed laminar matroid). It then adds the missing part of every set whose descendants cover some, but not all, of its elements. The following result shows that it does so while preserving feasibility.

\begin{restatable}{lemma}{PackingPreserving}\label{lem:packing-preserving}
    The laminar packing operation is feasibility-preserving.
\end{restatable}

\begin{proof}
    Suppose we apply laminar packing to an $n$-colorable matroid $(M,\calL,u)$ to produce a matroid $(M,\calL',u')$.

    First, we prove that $\calL'$ is a laminar set family. The intermediate family $\calL_0$ is laminar because $M$ contains every original set. Consider a missing part $S^*$ added for some $S \in \calL_0$ and any $T \in \calL_0$. If $T \subsetneq S$, then $T \in D_0(S)$, so $S^* \cap T = \emptyset$. If $S \subseteq T$, then $S^* \subseteq T$, and if $S \cap T = \emptyset$, then $S^* \cap T = \emptyset$. Thus, $S^*$ is laminar with every set in $\calL_0$. Finally, consider missing parts $S^*$ and $T^*$ added for two distinct sets $S,T \in \calL_0$. If $S$ and $T$ are disjoint, then so are $S^*$ and $T^*$. If, without loss of generality, $S \subsetneq T$, then $S \in D_0(T)$, so $T^*$ is disjoint from $S$ and hence from $S^*$. Therefore, $\calL'$ is laminar.

    Next, we prove the rank preservation condition, i.e., that $(M,\calL',u')$ is also $n$-colorable. It is well known that a laminar matroid $(M,\calL,u)$ is $n$-colorable if and only if $u(S) \geq \ceil{\nicefrac{|S|}{n}}$ for all $S \in \calL$ (for completeness, we provide a short proof of this fact in \Cref{lem:well-known-proof}). Since this condition still holds for $(M,\calL',u')$, it is also $n$-colorable.

    Finally, we prove the independent set shrinkage condition, i.e., $\calI(M,\calL',u') \subseteq \calI(M,\calL,u)$.  For contradiction, suppose there exists $I \in \calI(M,\calL',u') \setminus \calI(M,\calL,u)$. Since $I \notin \calI(M,\calL,u)$, there must exist an original set $S \in \calL$ for which $|I \cap S| > u(S)$. Every original set remains in $\calL'$, and the $n$-colorability of the original matroid implies $u(S) \ge \ceil{\nicefrac{|S|}{n}} = u'(S)$. Hence, $|I \cap S| > u'(S)$, contradicting the fact that $I \in \calI(M,\calL',u')$.
\end{proof}

\subsection{Stage 2: Simplification}\label{sec:simplification}

Once the input laminar matroid is packed in Stage 1, we apply several operations on the resulting matroid, which are not only feasibility-preserving, but also keep the resulting matroid packed.
\begin{definition}[Packing-Preserving Operation]
    We say that an operation that takes a packed laminar matroid $(M,\calL,u)$ as input and returns a laminar matroid $(M,\calL',u')$ as output is \emph{packing-preserving} if $(M,\calL',u')$ is also packed.
\end{definition}

We will introduce four such operations: \emph{splitting}, \emph{child pruning}, \emph{parent pruning}, and \emph{partition extraction}. All operations will take an $n$-colorable laminar matroid $(M,\calL,u)$ as input and return a matroid $(M,\calL',u')$. To prove that they are feasibility-preserving and packing-preserving, we will need to show that they return a laminar matroid that is packed, $n$-colorable, and satisfies $\calI(M,\calL',u') \subseteq \calI(M,\calL,u)$. Our algorithm will successively keep applying these operations until none of them can be applied anymore.

\begin{definition}[Splitting]
    If there exists an $S \in \calL$ and $S_1 \subsetneq S$ satisfying:
    \begin{enumerate}
        \item[(i)] $|S_1|$ is a positive multiple of $n$; and
        \item[(ii)] for each $C \in C(S)$, either $C \subseteq S_1$ or $C \cap S_1 = \emptyset$,
    \end{enumerate}
    then replace $S$ with two disjoint sets $S_1$ and $S_2 = S \setminus S_1$ (i.e., $\calL' = \calL \cup \set{S_1,S_2} \setminus \set{S}$). Set $u'(S_1) = \nicefrac{|S_1|}{n}$, $u'(S_2) = \ceil{\nicefrac{|S_2|}{n}}$, and $u'(T) = u(T)$ for every $T \in \calL' \setminus \set{S_1,S_2}$.
\end{definition}
In the tree form, this amounts to removing node $S$ and replacing it with the two sets $S_1$ and $S_2$ that partition it (if one of these sets already exists in $\calL$, then we simply leave it there). Every other former child $C$ of $S$ (and the subtree below it) is placed under the unique replacement set that contains it.

\begin{lemma}\label{lem:splitting}
    The splitting operation is feasibility-preserving and packing-preserving.
\end{lemma}
\begin{proof}
    First, we show that the resulting matroid remains laminar. Since the laminar structure holds for all the pairs of sets that were part of $\calL$, we simply need to establish it for pairs of sets in which one or both sets come from $\set{S_1,S_2}$. It holds for the $(S_1,S_2)$ pair as $S_1 \cap S_2 = \emptyset$. Take the pair $(S_1,T)$ for some $T \in \calL \setminus \set{S}$. If $T \in D(S)$, then by definition we have either $T \subseteq S_1$ or $T \subseteq S_2$ (in which case $T \cap S_1 = \emptyset$). If $S \subseteq T$, then $S_1 \subseteq T$. Finally, if $S \cap T = \emptyset$, then $S_1 \cap T = \emptyset$. Hence, the laminar structure holds for all pairs $(S_1,T)$ where $T \in \calL \setminus \set{S}$. A symmetric argument holds when replacing $S_1$ with $S_2$. Hence, $\calL'$ remains laminar.

    Next, we show that $\calL'$ remains packed. First, suppose that $S$ is a leaf. Since $S_1$ and $S_2$ are nonempty strict subsets of $S$, neither can already belong to $\calL$; otherwise, it would be a descendant of $S$. Thus, both are new leaves in $\calL'$ and satisfy packedness. Now suppose that $S$ is not a leaf, and fix $j \in \set{1,2}$. If $S_j$ already belongs to $\calL$, removing its strict ancestor $S$ does not change the descendants of $S_j$. Hence, $S_j$ retains its original packedness. If $S_j$ is new, then the children of $S$ contained in $S_j$ become descendants of $S_j$. Since the children of $S$ partition $S$, and each child is contained in exactly one of $S_1$ and $S_2$, these children cover $S_j$. Therefore, $S_j$ is packed. Finally, for every strict ancestor $T$ of $S$, replacing the descendant $S$ with $S_1$ and $S_2$ leaves the union of the descendants of $T$ unchanged because $S_1 \cup S_2=S$. Thus, $\calL'$ is packed.

    We now show that $(M,\calL',u')$ remains $n$-colorable. For every $T \in \calL' \setminus \set{S_1,S_2}$, because the input matroid is $n$-packed, we have $u'(T)=u(T)=\ceil{\nicefrac{|T|}{n}}$. Since $|S_1|$ is a multiple of $n$, we have $u'(S_1)=\nicefrac{|S_1|}{n}=\ceil{\nicefrac{|S_1|}{n}}$, while the definition of splitting gives $u'(S_2)=\ceil{\nicefrac{|S_2|}{n}}$. In particular, if either $S_j$ already belonged to $\calL$, its prescribed capacity agrees with its original capacity because the input matroid is $n$-packed. Thus, $u'(T) \geq \ceil{\nicefrac{|T|}{n}}$ for every $T \in \calL'$, so \Cref{lem:well-known-proof} implies that $(M,\calL',u')$ is $n$-colorable.

    Finally, we show independent set shrinkage: $\calI(M,\calL',u') \subseteq \calI(M,\calL,u)$. Suppose for contradiction that there exists $I \in \calI(M,\calL',u') \setminus \calI(M,\calL,u)$. Then, there must be some $T \in \calL$ such that $|I \cap T| > u(T)$. Since the only set in $\calL \setminus \calL'$ is $S$, it must be the case that $T = S$. However, since $I \in \calI(M,\calL',u')$, we have that:
    \begin{align*}
        |I \cap S_1| &\leq u'(S_1) = \nicefrac{|S_1|}{n}, \\
        |I \cap S_2| &\leq u'(S_2) = \ceil{\nicefrac{|S_2|}{n}}.
    \end{align*}
    Summing these gives:
    \begin{align*}
        |I \cap S| &= |I \cap S_1| + |I \cap S_2| \\
        &\leq \nicefrac{|S_1|}{n} + \ceil{\nicefrac{|S_2|}{n}}\\
        &=\ceil{\nicefrac{(|S_1|+|S_2|)}{n}} =\ceil{\nicefrac{|S|}{n}}.
    \end{align*}
    This contradicts $|I \cap S| > u(S)$. Crucially, the penultimate step holds because $\nicefrac{|S_1|}{n}$ is an integer.
\end{proof}

\begin{definition}[Child Pruning]
    If there exists $S \in \calL$ with $D(S) \neq \emptyset$ and $|S| \leq n$, remove all descendants of $S$ (i.e., set $\calL' = \calL \setminus D(S)$ and keep $u'(T) = u(T)$ for all $T \in \calL'$).
\end{definition}
In the tree form, this amounts to deleting the nonempty part of the tree below the non-leaf node $S$, thereby making $S$ a leaf.

\begin{lemma}\label{lem:child-pruning}
    The child pruning operation is feasibility-preserving and packing-preserving.
\end{lemma}
\begin{proof}
    Since this operation only removes sets, $\calL'$ remains laminar and $(M,\calL',u')$ remains $n$-colorable. Because all the descendants of $S$ are removed, $S$ becomes a leaf. For every strict ancestor $T$ of $S$, the descendant family loses the sets in $D(S)$ but retains $S$ itself, so the union of the descendants of $T$ remains unchanged. The descendants of every other remaining set are unaffected. Hence, the resulting matroid remains packed.

    Thus, we just need to show independent set shrinkage: $\calI(M,\calL',u') \subseteq \calI(M,\calL,u)$. Suppose for contradiction that there exists $I \in \calI(M,\calL',u') \setminus \calI(M,\calL,u)$. Then, there must be some $T \in \calL$ such that $|I \cap T| > u(T)$. Since $I \in \calI(M,\calL',u')$, we must have $T \in D(S)$, which implies $T \subsetneq S$. Given $|S| \leq n$, we get $0 < |T| < n$, so $u(T) = \ceil{\nicefrac{|T|}{n}} = 1$ and $u(S) = 1$. But then $|I \cap T| > 1$ also implies $|I \cap S| > 1$, contradicting $I \in \calI(M,\calL',u')$.
\end{proof}

\begin{definition}[Parent Pruning]
    If there exists $S \in \calL$ for which $C(S) \neq \emptyset$ and $\sum_{T \in C(S)} u(T) \leq u(S)$, then remove $S$ from $\calL$ (i.e., $\calL' = \calL \setminus \set{S}$ and $u'(T) = u(T)$ for all $T \in \calL'$).
\end{definition}
In the tree form, this amounts to deleting the non-leaf node $S$ and placing every child of $S$ directly under the parent of $S$; if $S$ is top-level, its children become top-level.

\begin{lemma}
    The parent pruning operation is feasibility-preserving and packing-preserving.
\end{lemma}
\begin{proof}
    As was the case in the previous proofs, since this operation only removes sets, the resulting matroid remains laminar and $n$-colorable. It is also packed. Because $C(S) \neq \emptyset$, the packedness of the original matroid implies that $C(S)$ is a disjoint partition of $S$. If $S$ has a parent $T$, then $\bigcup C(T) = T$ continues to hold when we replace $S$ in $C(T)$ with the sets in $C(S)$. If $S$ is top-level, its children simply become top-level, and the packedness of each remaining set is unchanged.

    Thus, we just need to show independent set shrinkage: $\calI(M,\calL',u') \subseteq \calI(M,\calL,u)$. Suppose for contradiction that there exists $I \in \calI(M,\calL',u') \setminus \calI(M,\calL,u)$. Then, there must be some $T \in \calL$ such that $|I \cap T| > u(T)$. Since only $S$ was removed, we must have $T=S$, so $|I \cap S| > u(S)$. But for all $C \in C(S)$, we still have $|I \cap C| \leq u(C)$. Summing over the children, we get
    \[
        |I \cap S| = \sum_{C \in C(S)} |I \cap C| \leq \sum_{C \in C(S)} u(C) \leq u(S),
    \]
    where the first transition uses $C(S) \neq \emptyset$ and the fact that, by packedness, $C(S)$ is a disjoint partition of $S$, and the last transition comes from the definition of the operation. This is a contradiction.
\end{proof}

\begin{definition}[Partition Extraction]
    If there exists $S \in \calL$ with $D(S) = \emptyset$ (leaf), $|S| = n$, and $S \subsetneq T$ for some $T \in \calL$, remove the elements of $S$ from every strict superset of $S$ (i.e., $\calL' = \set{S} \cup \set{T \setminus S : T \in \calL \setminus \set{S}}$). Set $u'(T')=\ceil{\nicefrac{|T'|}{n}}$ for every $T'\in\calL'$.
\end{definition}
In the tree form, this operation detaches a leaf $S$ of size $n$ from its ancestors and leaves it as a singleton component in the tree. The strict-superset condition ensures that the operation changes at least one set rather than re-extracting an $S$ that is already isolated.

\begin{lemma}
    The partition extraction operation is feasibility-preserving and packing-preserving.
\end{lemma}
\begin{proof}
    First, we show that $\calL'$ remains laminar. We need to show that for any $T'_1,T'_2 \in \calL'$, $T'_1 \cap T'_2 = \emptyset$, $T'_1 \subseteq T'_2$, or $T'_2 \subseteq T'_1$. If one of the sets is $S$, then we have an empty intersection because the elements of $S$ are removed from all other sets in $\calL'$. If neither set is $S$, then $T'_1 = T_1 \setminus S$ and $T'_2 = T_2 \setminus S$ for $T_1,T_2 \in \calL$. Since $\calL$ is laminar, we have that $T_1 \cap T_2 = \emptyset$ (in which case $(T_1 \setminus S) \cap (T_2 \setminus S) = \emptyset$), $T_1 \subseteq T_2$ (in which case $T_1 \setminus S \subseteq T_2 \setminus S$), or $T_2 \subseteq T_1$ (in which case $T_2 \setminus S \subseteq T_1 \setminus S$). Hence, the laminar structure is established for $\calL'$.

    Next, we prove that $(M,\calL',u')$ is $n$-packed. The retained set $S$ is a leaf in $\calL'$. Consider any $T'\in\calL'\setminus\set{S}$. If $T'$ coincides with an original set $U\in\calL$ disjoint from $S$, then the descendants of $U$ are unchanged: its original descendants are disjoint from $S$, while laminarity prevents $U$ from obtaining any new descendants through this operation, as such a set from $\calL$ causing a new descendant after removing $S$ would need to contain $S$ and some, but not all, of $U$. Thus, $U$ retains its original packedness. Otherwise, $T'=T\setminus S$ is genuinely new for some strict ancestor $T$ of $S$. Since $\calL$ is packed, $\bigcup_{X\in D(T)}X=T$, and the transformed descendants cover
    \[
        \bigcup_{X\in D(T)\setminus\set{S}}(X\setminus S)
        =
        \left(\bigcup_{X\in D(T)}X\right)\setminus S
        =
        T\setminus S
        =
        T'.
    \]
    Hence, $T'$ is packed, and therefore $\calL'$ is packed.

    The capacity condition for being $n$-packed holds directly from the definition of $u'$. If $T\cap S=\emptyset$, then $T'=T$ and $u'(T')=\ceil{\nicefrac{|T|}{n}}=u(T)$. If $S\subsetneq T$, then $T'=T\setminus S$, $|T'|=|T|-n$, and
    \[
        u'(T')
        =
        \ceil{\nicefrac{|T|-n}{n}}
        =
        \ceil{\nicefrac{|T|}{n}}-1
        =
        u(T)-1.
    \]
    Notably, if $T\setminus S$ coincides with an existing set $U$ from $\calL$, the packedness of the original matroid will ensure that it already has a capacity of $u(U)=u(T)-1$. Therefore, the resulting matroid is $n$-packed.

    Thus, we just need to show independent set shrinkage: $\calI(M,\calL',u') \subseteq \calI(M,\calL,u)$. Suppose for contradiction that there exists $I \in \calI(M,\calL',u') \setminus \calI(M,\calL,u)$. Since the only changed constraints correspond to strict supersets of $S$, there exists some $T \in \calL$ such that $S \subsetneq T$, $|I \cap T| > u(T)$, and $|I \cap T'| \le u'(T') = u(T)-1$ for $T' = T \setminus S$. However, this implies $|I \cap S| > 1$. But since $S \in \calL$, $(M, \calL, u)$ is $n$-packed, and $u(S) = 1$, this is impossible, completing the contradiction.
\end{proof}

\subsection{Stage 3: Decomposition}\label{sec:decomposition}
Next, we show that the laminar matroid that remains at the end of Stage 2 has a special structure: it forms a forest in which every connected component is either a uniform matroid or a special laminar matroid over $6$ goods, which we term a \emph{key matroid}. Let us formalize this.

\begin{definition}[Connected Components]
Define the set of top-level (i.e., inclusion-wise maximal) sets in $\calL$ as $\top(\calL) = \set{S \in \calL : \nexists T \in \calL, S \subsetneq T}$. Define the \emph{connected component} induced by each $S \in \top(\calL)$ to be the matroid $(M',\calL',u')$, where $M' = S$, $\calL' = \set{S} \cup D(S)$, and $u'(T) = u(T)$ for all $T \in \calL'$.
\end{definition}

Note that when $D(S) = \emptyset$, the connected component has a singleton $\calL' = \set{S}$, which induces a uniform matroid. We will prove that the only other possibility is the following.

\begin{definition}[Key Matroid]
    We say that a laminar matroid $(M,\calL,u)$ is a \emph{key matroid} if:
    \begin{itemize}
        \item $|M| = 6$,
        \item $|\calL| = 4$,
        \item $M \in \calL$ with $u(M) = 2$, and
        \item the remaining three sets in $\calL$ have size $2$ each, form a disjoint partition of $M$, and have an upper bound of $1$ each.
    \end{itemize}
\end{definition}

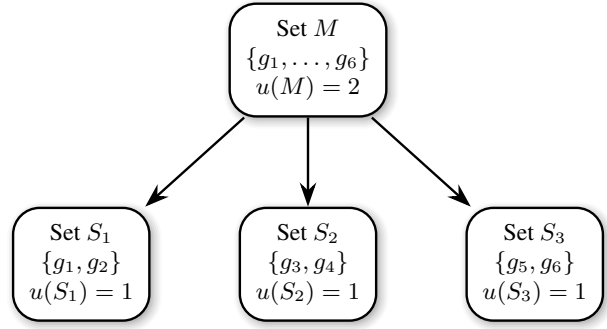
\begin{figure}[htb!]
    \centering
    \begin{tikzpicture}[
        level 1/.style={sibling distance=3.0cm},
        level distance=2.7cm,
        edge from parent/.style={
            draw,
            line width=0.9pt,
            rounded corners=6pt,
            -{Stealth[length=3.2mm,width=2.2mm]},
        },
        snode/.style={
            rectangle,
            rounded corners=10pt,
            draw=black,
            line width=0.9pt,
            align=center,
            font=\small,
            fill=white,
            inner sep=6pt,
            blur shadow={shadow blur steps=5,shadow blur radius=3.5pt,shadow xshift=1.2pt,shadow yshift=-1.2pt}
        },
    ]
        \node[snode] (root) {Set $M$ \\ $\{g_1, \ldots, g_6\}$ \\ $u(M)=2$}
            child { node[snode] {Set $S_1$ \\ $\{g_1, g_2\}$ \\ $u(S_1)=1$} }
            child { node[snode] {Set $S_2$ \\ $\{g_3, g_4\}$ \\ $u(S_2)=1$} }
            child { node[snode] {Set $S_3$ \\ $\{g_5, g_6\}$ \\ $u(S_3)=1$} };
    \end{tikzpicture}
    \caption{Structure of the key matroid}
    \label{fig:key_matroid_figure}
\end{figure}

\begin{lemma}\label{lem:reduction}
    Let $(M,\calL,u)$ be a $3$-packed laminar matroid. If none of the splitting, child pruning, parent pruning, or partition extraction operations can be applied to $(M,\calL,u)$, then every connected component of $(M,\calL,u)$ is either a uniform matroid or a key matroid.
\end{lemma}
\begin{proof}
Fix a connected component $(M',\calL',u')$.
Assume it is neither a uniform matroid nor a key matroid; we derive a contradiction.

Throughout, since $(M,\calL,u)$ is $3$-packed, we use $u(X)=\lceil |X|/3\rceil$ for every $X\in\calL$.

\medskip
\noindent\textbf{Step 1: Bounding leaf sizes.}
Let $S\in\calL'$ be a leaf, i.e., $D(S)=\emptyset$.
\begin{itemize}
    \item If $|S|=3$, then \emph{partition extraction} applies to $S$ (since the component is not a uniform matroid, $S$ cannot simultaneously be a leaf and a top-level set, so it must have the parent set required for this operation). This contradicts the fact that no operation applies.
    \item If $|S|>3$, then \emph{splitting} applies to $S$: since $S$ is a leaf, the condition involving $C(S)$ is vacuous, and we may choose any subset $S_1\subsetneq S$ with $|S_1|=3$ (a positive multiple of $3$), resulting in a contradiction.
\end{itemize}
Hence, every leaf $S$ satisfies $|S|\le 2$.

\medskip
\noindent\textbf{Step 2: Locating key matroids.}
Choose a non-leaf set $S\in\calL'$ that is \emph{lowest} such that every $C \in C(S)$ is a leaf. Given the observation in Step 1, we have $|C| \le 2$ for all $C \in C(S)$.

\smallskip
\noindent\emph{(i) Lower bound on $|S|$.}
If $|S|\le 3$, then \emph{child pruning} applies to $S$, contradicting the fact that no operation applies. Therefore $|S|>3$.

\smallskip
\noindent\emph{(ii) No singleton child.}
Since $\calL$ is packed, $C(S)$ forms a disjoint partition of $S$ (\Cref{lem:packed-children-partition}). If some child $C\in C(S)$ has $|C|=1$, then regardless of whether there exists another child $C' \in C(S)$ with $|C'|=2$ or all other children $C' \in C(S)$ have $|C'|=1$, there would always exist a subcollection of children $\mathcal{C} \subseteq C(S)$ whose union has size $|\bigcup \mathcal{C}| = 3$ (either one $2$-set plus one $1$-set, or three $1$-sets). Let $S_1 = \bigcup \mathcal{C}$ with $|S_1|=3$. Since $S_1$ is a union of children of $S$, every child of $S$ is either contained in $S_1$ or disjoint from it. Hence \emph{splitting} applies to $S$, which is a contradiction. Thus, every child of $S$ has size exactly $2$.

Consequently, $|S|$ is even and at least $4$.

\medskip
\noindent\textbf{Step 3: Pinning down $|S|$.}
Write $|S|=2k$, where $k=|C(S)|\ge 2$.
We rule out all values except $|S|=6$.

\begin{itemize}
    \item If $|S|=4$ (so $k=2$), then $u(S)=\lceil 4/3\rceil=2$ and each child $C$ has $u(C)=\lceil 2/3\rceil=1$. Hence
    \[
        \sum_{C\in C(S)} u(C)=2 \le u(S),
    \]
    so \emph{parent pruning} applies to $S$, a contradiction.

    \item If $|S|\ge 8$ (so $k\ge 4$), pick any three children $C_1,C_2,C_3\in C(S)$ and let $S_1=C_1\cup C_2\cup C_3$. Then $|S_1|=6$ is a positive multiple of $3$, and again $S_1$ is a union of children, so \emph{splitting} applies to $S$, which is a contradiction.
\end{itemize}
Therefore, $|S|=6$ and $C(S)$ consists of exactly three children of size $2$ each.

\medskip
\noindent\textbf{Step 4: The component must be exactly the key matroid.}
If $S$ is a top-level set of the component (i.e., $S\in\top(\calL)$), then the component consists of $S$ together with its three children. With $u(S)=\ceil{\nicefrac{6}{3}}=2$ and each child having capacity $\ceil{\nicefrac{2}{3}}=1$, this is precisely the key matroid, contradicting our assumption that the component is not a key matroid.

Otherwise, $S$ has a parent $T$ in the component, i.e., $S\in C(T)$. Let $S_1:=S$. Then $|S_1|=6$ is a positive multiple of $3$, and, since $S_1$ is itself a child of $T$, every child of $T$ is either contained in $S_1$ (namely $S$) or disjoint from it (all other children). Hence \emph{splitting} applies to $T$, which is a contradiction.

Thus, we reach a contradiction in all cases, meaning that every connected component must be either a uniform matroid or a key matroid.
\end{proof}

\subsection{Stage 4: Solving the Key Matroid}\label{sec:key}

Finally, we ``solve'' the key matroid. As we will see in \Cref{sec:algo}, our final algorithm requires an EF1 allocation with an additional \emph{ordered no-envy} condition subject to the key matroid: for a given agent ordering, we need to be able to prevent envy altogether from each agent towards any agent later in the ordering.

\begin{restatable}{lemma}{EFOneKeyMatroid}\label{lem:EF1-key-matroid}
    Given $n=3$ agents with additive valuations over $m=6$ goods, a key laminar matroid, and an ordering over the agents, there always exists a complete and feasible EF1 allocation in which no agent envies any agent later in the ordering.
\end{restatable}

\begin{proof}
    Without loss of generality, fix the agent ordering for the ordered no-envy condition such that we want agent $i$ to not envy agent $j$ for any $i < j$.

    We construct such an allocation by partitioning the goods into $3$ bundles, letting agent $1$ pick her favorite bundle, letting agent $2$ pick her favorite bundle among those remaining, and giving the final bundle to agent $3$. This picking order guarantees the ordered no-envy condition. To ensure EF1, we design the partition to satisfy the following conditions:
    \begin{enumerate}
        \item The key matroid constraints: each bundle has size $2$, and (without loss of generality) the two goods in any of the sets $\set{g_1,g_2}, \set{g_3,g_4}, \set{g_5,g_6}$ lie in different bundles.
        \item The two most favorite goods of agent $2$, say $\set{g_a,g_b}$, lie in different bundles.
        \item The three most favorite goods of agent $3$, say $\set{g_x,g_y,g_z}$, lie in different bundles.
    \end{enumerate}

    First, let us see why this ensures EF1. Since each agent receives exactly two goods, EF1 reduces to ensuring that no agent $i$ strictly prefers her less-preferred good from $A_j$ to the bundle of two goods $A_i$ that she receives. Agent $1$ already does not envy anyone. For agent $2$, her value for any bundle after removing her more valuable good in that bundle is at most $\min\set{v_2(g_a),v_2(g_b)}$. However, since $g_a$ and $g_b$ are in different bundles and she goes second in the picking order, her value for her own bundle is at least $\min\set{v_2(g_a),v_2(g_b)}$, ensuring EF1. The same reasoning applies to agent $3$: her value for her own bundle is at least $\min\set{v_3(g_x),v_3(g_y),v_3(g_z)}$, while her value for any bundle after removal of her favorite good from that bundle is at most this value.

    It remains to show that a partition of the $6$ goods satisfying the three conditions above always exists. Note that a partition into three bundles of size $2$ each can be depicted as a perfect matching in an undirected graph over the goods. The third condition---the goods in $T = \set{g_x,g_y,g_z}$ lie in different bundles---reduces to seeking this perfect matching in the bipartite graph between $T$ and $M \setminus T$. Finally, the first two conditions---the two goods in any of the sets $\set{g_1,g_2}$, $\set{g_3,g_4}$, $\set{g_5,g_6}$, and $\set{g_a,g_b}$ lie in different bundles---simply forbid the four edges between the four pairs of goods.

    Starting with the complete bipartite graph $K_{3,3}$, first remove the three edges $\set{g_1,g_2}$, $\set{g_3,g_4}$, and $\set{g_5,g_6}$. In the worst case, all three edges go across the two parts and get removed, leaving a graph with each vertex having degree $2$, which is a $6$-cycle (if not all three of those edges go across the parts, the remaining graph is a superset of a $6$-cycle). This can be decomposed into two perfect matchings, so even after removing the edge $\set{g_a,g_b}$, the perfect matching not containing this edge remains and can be selected.
\end{proof}

\subsection{Final Algorithm}\label{sec:algo}

We are now ready to present our full (polynomial-time) algorithm for computing an EF1 allocation for three agents subject to a laminar matroid. The algorithm first packs the given laminar matroid, then applies the operations from \Cref{sec:simplification}
until none of those operations can be applied anymore. \Cref{lem:reduction} shows that each connected component of the resulting matroid is either a uniform matroid or a key matroid. Our algorithm allocates the goods one connected component at a time.

In order to maintain EF1, it uses the same trick that \citet{biswas2018fair} use for computing an EF1 allocation subject to a partition matroid, where the connected components are simply uniform matroids. They maintain the property that, at any given time, in addition to the partial allocation being EF1, there is a ``no-envy order'' among the agents such that agents earlier in the order do not envy agents later in the order. Then, the next connected component is allocated using the round-robin method~\cite{CKMP+19} with the picking order being the reverse of the no-envy order. This ensures that envy does not compound and the resulting allocation remains EF1. The no-envy order property is then restored by running the envy-cycle elimination procedure of \citet{LMMS04}. The only difference in our algorithm is to use \Cref{lem:EF1-key-matroid} instead of the round-robin method when the connected component is a key matroid. The heavy lifting of our technique is already done by \Cref{lem:reduction}, which establishes the structure of the connected components, and \Cref{lem:EF1-key-matroid}, which solves the key matroid.

\begin{proof}[Proof of \Cref{thm:n=3}]
    \Cref{lem:reduction} proves that by \Cref{line:end-while} of \Cref{alg:EF1-alg}, the connected components of the matroid $(M,\calL,u)$ are uniform matroids and key matroids.

    We will first show that the allocation resulting from \Cref{alg:EF1-alg} will be a complete and feasible allocation with respect to the original laminar matroid. From our proofs that all the reduction operations used by \Cref{alg:EF1-alg} are feasibility-preserving, it is sufficient to prove that the final allocation is complete and feasible with respect to the reduced laminar matroid produced after \Cref{line:end-while}.

    First, note that, in the reduced laminar matroid produced by the algorithm, all items will be included in some connected component. Laminar packing places $M$ in the laminar family, so immediately after packing, $M$ is the unique top-level set. Each reduction operation preserves the property that the top-level sets form a disjoint partition of $M$: splitting replaces a set by two sets that partition it; child pruning retains $S$; parent pruning either leaves the top-level sets unchanged or replaces a top-level $S$ by its children, which always fully partition $S$ due to the laminar matroid being packed; and partition extraction retains $S$ while removing its goods from its strict supersets. Consequently, the connected components at the end of Stage 2 form a disjoint partition of $M$.

    From this, it is easy to see that the allocation step of the algorithm will allocate all the goods. For each key-matroid component, the procedure from \Cref{lem:EF1-key-matroid} assigns all six goods among the agents, while the round-robin procedure allocates every good in each uniform-matroid component. Finally, the envy-cycle elimination step only shifts bundles between agents; it never leaves an allocated good unassigned. Thus, the algorithm always produces a complete allocation of all goods in $M$.

    It remains to show that the final allocation is feasible with respect to the reduced laminar matroid produced in \Cref{alg:EF1-alg}. The goods in each connected component are allocated feasibly with respect to that component. Since every set in the reduced laminar family belongs to exactly one connected component, satisfying the constraints of every connected component is equivalent to satisfying the constraints of the entire laminar matroid. For each key-matroid component, feasibility follows from \Cref{lem:EF1-key-matroid}. For a uniform-matroid component over a set of goods $S \subseteq M$, no agent may receive more than $\ceil{\nicefrac{|S|}{n}}$ goods from $S$. The round-robin procedure produces a balanced allocation, so every agent respects this capacity. Finally, envy-cycle elimination only shifts fixed bundles between agents. Since all agents are subject to the same constraints, reassigning these bundles does not create a new feasibility violation. Thus, the final allocation is feasible.

    To prove that the final allocation is EF1, we inductively show that the following hold each time the algorithm reaches \Cref{line:alloc-start}:
    \begin{itemize}
        \item $A$ is EF1, and
        \item $O$ is the topological envy order of $A$, i.e., under $A$, no agent envies another agent who appears earlier in the ordering $O$.
    \end{itemize}
    Clearly, this is true in the beginning when the algorithm first arrives at \Cref{line:alloc-start} with $A=\emptyset$. Suppose it is true when the algorithm arrives at \Cref{line:alloc-start} with current allocation $A$ and appends the allocation $A^{\calC}$ of the next connected component $\calC$ to $A$.

    Note that $A^{\calC}$ is EF1 and $O$ is its reverse topological envy order. When $\calC$ is a uniform matroid, this is a well-known property of the round-robin method~\cite{CKMP+19}. When $\calC$ is a key matroid, this is established in \Cref{lem:EF1-key-matroid}. Finally, it is also known in the literature that taking the union of two EF1 allocations whose topological envy orders are the reverse of each other yields an EF1 allocation (e.g., see \citet[Lemma 2]{biswas2018fair}). Hence, $A$ is still EF1 at the end of \Cref{line:union}. Then, after applying envy-cycle elimination in \Cref{line:envy-cycle}, $A$ remains EF1 and $O$ is set to be its new topological envy order, proving the induction step.

    Finally, we can prove that \Cref{alg:EF1-alg} will terminate in polynomial time. Since $n=3$ is fixed, the four simplification operations can be detected and performed in polynomial time: the pruning and extraction operations require only scanning $\calL$, while splitting can be found by checking all nonempty proper subcollections of at most three children (or choosing any three goods when $S$ is a leaf). All remaining steps of \Cref{alg:EF1-alg} are also executable in polynomial time. Hence, it just remains to show that the main while loop of the algorithm will take polynomially many iterations prior to termination.

    It is well known that any laminar set family $\calL$ over a ground set $M$ contains at most $2|M| - 1$ nonempty sets. Therefore, for any such $\calL$, we must have that $\sum_{S \in \calL}|S| \leq |M|(2|M| - 1) \in O(|M|^2)$.

    Note that each of the operations we considered, when run against a laminar matroid $(M,\calL,u)$, can only weakly decrease $\sum_{S \in \calL}|S|$, and if an operation does not strictly decrease this value, then it strictly increases the value $|\calL|$.

    Each child pruning and parent pruning operation removes at least one nonempty set from $\calL$, strictly decreasing $\sum_{S \in \calL}|S|$. Each partition extraction operation removes the contents of its chosen set $S$ from every strict superset of $S$. Since $S$ is nonempty and one such superset must exist by definition, this decreases the sum objective as well. Finally, splitting replaces $S$ with the disjoint sets $S_1$ and $S_2$, whose sizes sum to $|S|$. If both sets are new to $\calL$, $\sum_{S \in \calL}|S|$ is unchanged while $|\calL|$ increases by one. If either set is already in $\calL$, it is not added a second time, and the sum of set sizes strictly decreases instead.

    Since $\sum_{S \in \calL}|S|$ is a nonnegative integer of value $O(|M|^2)$ and is never increased by any operation, it can strictly decrease at most $O(|M|^2)$ times. Between two consecutive strict decreases of the sum, every iteration strictly increases $|\calL|$. Since $|\calL|=O(|M|)$, there can be at most $O(|M|)$ such iterations between consecutive strict decreases. Therefore, the while loop can run for at most $O(|M|^3)$ iterations.
\end{proof}

\begin{algorithm}[H]
\caption{3-Agent EF1 for Laminar Matroids}\label{alg:EF1-alg}
\DontPrintSemicolon
\SetAlgoLined

\SetKwInput{KwInput}{Input}

\KwInput{Additive valuations $(v_1,v_2,v_3)$, $3$-colorable laminar matroid $(M,\mathcal{L},u)$}

$(M,\calL,u) \gets$ Apply laminar packing to $(M,\calL,u)$\;

\While{any operation \texttt{op} $\in$ \{splitting, child pruning, parent pruning, partition extraction\} can be applied to $(M,\mathcal{L},u)$}{
    $(M,\calL,u) \gets$ Apply \texttt{op} to $(M,\calL,u)$\;
}

$O \gets (1,2,3)$\label{line:end-while}\;
$A \gets (A_1=\emptyset,\ A_2=\emptyset,\ A_3=\emptyset)$\;

\For{each connected component $\calC$ of $(M,\calL,u)$}{\label{line:alloc-start}
    \uIf{$\calC$ is a uniform matroid}{
        $A^{\calC} \gets$ Allocate the goods in $\calC$ via the round-robin method with picking order $O$\;
    }
    \Else{
        $A^{\calC} \gets$ Allocate the goods in $\calC$ via \Cref{lem:EF1-key-matroid} with the no-envy order $O$\;
    }

    $A \gets A \cup A^{\calC}$\label{line:union}\;
    $A \gets$ Apply envy-cycle elimination of \citet{LMMS04} to $A$\label{line:envy-cycle}\;
    $O \gets$ topological order of the agents from the envy graph (so no agent envies any agent earlier in the ordering $O$)\;
}\label{line:alloc-end}

\Return $A$\;
\end{algorithm}

\section{Matroids With More Agents}\label{sec:partition-reductions}

We achieve our main result by decomposing any laminar matroid into a new matroid that is \emph{almost} a partition matroid. If there were a way to continue the decomposition process, breaking the key matroids down into uniform matroids as well, then we would no longer need our additional proofs on how to solve key matroids. We could simply use the algorithm of \citet{biswas2018fair} without any modifications on the fully decomposed matroid to get a solution. However, it is not hard to confirm that for the key matroid $(M,\calL,u)$, there does not exist any $3$-colorable partition matroid $(M,\calI)$ such that $\calI \subseteq \calI(M,\calL,u)$, so there is no way to continue this decomposition process while maintaining $3$-colorability and the Independent Set Shrinkage property, which are both vital for our proof.

Interestingly, the work of \citet{berczi2021list} also studies the technique of taking a complex matroid, reducing it to a partition matroid that meets the Independent Set Shrinkage property, and then solving their problem on that simpler partition matroid (the problem they study is list-coloring two matroids). Their reduction may increase the number of colors required by up to roughly a factor of $2$. With this relaxation, they are able to reduce their matroids to \emph{exactly} partition matroids, without any pesky key matroid-like pieces left lying around. Below we state their result formally.

\begin{restatable}{theorem}{BercziReduction}\label{thm:berczi-reduction}
    For any $n'$-colorable matroid $(M, \mathcal{I})$, there exists a $k$-colorable partition matroid $(M, \mathcal{I}')$ such that $\mathcal{I}' \subseteq \mathcal{I}$, where $k$ depends on the class of $(M,\mathcal{I})$ as follows:
    \begin{enumerate}
        \item If $(M,\mathcal{I})$ is a transversal matroid, then $k = n'$ (Theorem 1.5 of \cite{berczi2021list}).
        \item If $(M,\mathcal{I})$ is a graphic matroid, then $k = 2n' - 1$ (Theorem 1.6 of \cite{berczi2021list}).
        \item If $(M,\mathcal{I})$ is a gammoid and $n' \geq 2$, then $k = 2n' - 2$ (Theorem 1.10 of \cite{berczi2021list}). When $n'=1$, $(M,\mathcal{I})$ is the free matroid and is itself a $1$-colorable partition matroid, so we take $k=1$.
    \end{enumerate}
\end{restatable}

It turns out that this style of reduction plus the partition matroid algorithm of \citet{biswas2018fair} also yields some very interesting results, which we illustrate below:

\begin{restatable}{lemma}{PartitionAgentReduction}\label{lem:partition-agent-reduction}
    Let $N=[n]$ be a set of agents with additive valuation profile $\vals$. Let $(M,\mathcal{I})$ be a matroid and let $(M,\mathcal{I}')$ be an $n$-colorable partition matroid such that $\mathcal{I}' \subseteq \mathcal{I}$. Then, there exists a complete and feasible EF1 allocation of $M$ to the agents in $N$ with respect to $(M,\mathcal{I})$.
\end{restatable}
\begin{proof}
    The algorithm of \citet{biswas2018fair} finds such an allocation $A$ with respect to $(M,\mathcal{I}')$. Since $\mathcal{I}' \subseteq \mathcal{I}$, every bundle in $A$ is also independent in $(M,\mathcal{I})$, so $A$ is feasible with respect to $(M,\mathcal{I})$ as well.
\end{proof}

Combining \Cref{thm:berczi-reduction} with \Cref{lem:partition-agent-reduction}, we obtain the following result:

\begin{restatable}{theorem}{ReductionResults}\label{thm:reduction-results}
    Let $N=[n]$ be a set of agents with additive valuation profile $\vals$. A complete and feasible EF1 allocation of $M$ to the agents in $N$ exists in each of the following cases:
    \begin{enumerate}
        \item $(M,\mathcal{I})$ is an $n$-colorable transversal matroid;
        \item $(M,\mathcal{I})$ is a $\lfloor(n+1)/2\rfloor$-colorable graphic matroid; or
        \item $(M,\mathcal{I})$ is a $\lfloor(n+2)/2\rfloor$-colorable gammoid.
    \end{enumerate}
\end{restatable}
\begin{proof}
    Apply \Cref{thm:berczi-reduction} to obtain a $k$-colorable partition matroid $(M,\mathcal{I}')$ such that $\mathcal{I}' \subseteq \mathcal{I}$. In the transversal case, take $n'=n$, which gives $k=n$. In the graphic case, take $n'=\lfloor(n+1)/2\rfloor$, which gives
    \[
        k=2\lfloor(n+1)/2\rfloor-1\leq n.
    \]
    In the gammoid case with $n\geq 2$, take $n'=\lfloor(n+2)/2\rfloor\geq 2$, which gives
    \[
        k=2\lfloor(n+2)/2\rfloor-2\leq n.
    \]
    If $n=1$, the assumed $1$-colorable gammoid is the free matroid and is itself a $1$-colorable partition matroid, so again $k=n$. Thus, in every case, $(M,\mathcal{I}')$ is $k$-colorable for some $k\leq n$, and is therefore also $n$-colorable. \Cref{lem:partition-agent-reduction} now yields the desired allocation of $M$ to the original $n$ agents with respect to $(M,\mathcal{I})$.
\end{proof}

Most notably, for $n$-colorable transversal matroids, \citet{berczi2021list} provide a perfect reduction to $n$-colorable partition matroids. Thus, if a transversal matroid is $n$-colorable, there always exists an EF1 allocation of its goods to $n$ agents.

Transversal matroids are able to model constrained fair division instances that require role fulfillment. A prime example is the fair allocation of individuals into balanced teams from a diverse talent pool. One can imagine dividing athletes into multiple sports teams: each athlete can play a specific set of positions, and every team must be constructed so that all required positions are filled by exactly one capable player. These perfectly formed teams represent the ``bases'' of a transversal matroid. Thus, using this algorithm, we would always be able to construct teams such that the coach of each team gets a collection of players they view as fair, while ensuring that all the teams are strictly functional on the field.

For graphic matroids and gammoids, \Cref{thm:reduction-results} proves a weaker version of the EF1 matroid conjecture. The population remains fixed at $n$ agents, but the theorem assumes that a graphic matroid is $\lfloor(n+1)/2\rfloor$-colorable or that a gammoid is $\lfloor(n+2)/2\rfloor$-colorable. The full conjecture would require only $n$-colorability, so resolving it for these classes requires weakening these stronger colorability assumptions.

\section{Discussion}

\paragraph{EF1 with ordered no-envy.} Our work suggests that the general approach of simplifying a given matroid into small connected components via feasibility-preserving operations and then allocating the connected components sequentially is promising for matroid-constrained fair division. However, this requires finding an EF1 allocation of each connected component that also has a directed acyclic envy graph with a given reverse topological order. Whether such allocations exist subject to all matroid constraints under which a complete allocation is feasible is a stronger question than ordinary EF1 existence. Such open questions have also been posed when seeking additional properties such as Pareto optimality~\cite{FSV20}.

\paragraph{Resolving the case of four agents.} In \Cref{app:four-agents}, we show that our technique reduces the four-agent problem under laminar matroids to finding EF1 allocations with the ordered no-envy condition in two small laminar matroids. Our efforts for a computer-aided resolution of this missing piece remained unfruitful, but this result may be within reach with a modest amount of effort, possibly even with a smarter brute-force strategy.

\paragraph{Matroid decompositions.} In \Cref{sec:partition-reductions}, we use partition decomposition bounds of \citet{berczi2021list} to establish EF1 guarantees under various classes of matroids subject to conditions stronger than $n$-colorability. \citet{berczi2021list} show that their bounds are tight, so improving them is not a viable approach to strengthening these results. However, our approach of decomposing into a partition matroid along with small other matroids seems to be a promising future avenue. In particular, can we design feasibility-preserving operations for broader classes of matroids than laminar matroids?

\section*{Acknowledgments}
This work was supported by an NSERC Discovery Grant and an NSERC-CSE Research Communities Grant. Researchers funded through the NSERC-CSE Research Communities Grants do not represent the Communications Security Establishment Canada or the Government of Canada. Any research, opinions or positions they produce as part of this initiative do not represent the official views of the Government of Canada.

\bibliographystyle{named}
\bibliography{abb,nisarg,ultimate}

\cleardoublepage

\appendix
\section*{Appendix}

\section{Omitted Proofs from \Cref{sec:section3}}\label{app:section3}

\subsection{Omitted Proofs from \Cref{sec:packing}}\label{app:packing}

\begin{lemma}\label{lem:well-known-proof}
    A laminar matroid $(M,\calL,u)$ is $n$-colorable if and only if $u(S) \geq \ceil{\nicefrac{|S|}{n}}$ for all $S \in \calL$.
\end{lemma}
\begin{proof}
    Necessity is simple. Clearly, to partition the elements of a set $S$ into $n$ disjoint subsets, some subset must contain at least $\ceil{\nicefrac{|S|}{n}}$ elements. Thus, if $u(S) < \ceil{\nicefrac{|S|}{n}}$, there is no way to partition the full set of items $M$ into $n$ disjoint subsets in a way that meets this constraint.

    To prove sufficiency, we appeal to a number of different theorems. First, it is well known that when a laminar set family is represented as a hypergraph (each element is a vertex and the edges are the sets in the family), then the incidence matrix of that hypergraph will be totally unimodular (see, for example,~\cite{biro2025stable}). Then, from~\cite{de1971equitable}, we know that for any hypergraph with such a property, and any $k \geq 1$, it is possible to partition the vertices of the hypergraph into $k$ disjoint sets, such that no edge $S$ of the hypergraph has more than $\ceil{\nicefrac{|S|}{k}}$ vertices in the same set. We can set $k = n$ and therefore get the existence of a valid $n$-coloring of our laminar matroid.
\end{proof}

\section{Four Agents}\label{app:four-agents}

In this section, we show which parts of our techniques can and cannot be applied to the case $n=4$ agents under laminar matroid constraints, both highlighting the versatility of our decomposition techniques and showing why EF1 under laminar matroid constraints beyond $3$ agents is a difficult problem. Recall that the packing and simplification steps in \Cref{sec:packing,sec:simplification}, respectively, hold for any number of agents. Only the steps to decompose into uniform and key matroids in \Cref{sec:decomposition} and solve the key matroid in \Cref{sec:key} are specific to the case of $n=3$. We show that similar steps can be executed for higher $n$, although it starts to get significantly trickier.

In \Cref{sec:decomposition}, we introduced the \emph{key matroid}, which was integral to our decomposition process for $3$-colorable matroids. In the case of $4$-colorable matroids, we take a similar approach. However, in this case, we must introduce two key matroids rather than one.

\begin{definition}[Size-$8$ Key Matroid]
    We say that a laminar matroid $(M,\calL,u)$ is a \emph{size-$8$ key matroid} if:
    \begin{itemize}
        \item $|M| = 8$,
        \item $|\calL| = 4$,
        \item $M \in \calL$ with $u(M) = 2$, and
        \item the remaining three sets in $\calL$ form a disjoint partition of $M$, have sizes $3$, $3$, and $2$, and each has an upper bound of $1$.
    \end{itemize}
\end{definition}

\begin{definition}[Size-$12$ Key Matroid]
    We say that a laminar matroid $(M,\calL,u)$ is a \emph{size-$12$ key matroid} if:
    \begin{itemize}
        \item $|M| = 12$,
        \item $|\calL| = 5$,
        \item $M \in \calL$ with $u(M) = 3$, and
        \item the remaining four sets in $\calL$ all have size $3$, form a disjoint partition of $M$, and each has an upper bound of $1$.
    \end{itemize}
\end{definition}

\begin{figure}[h]
    \centering
    \newcommand{\commontreescale}{0.65}

    \begin{minipage}{0.48\textwidth}
        \centering
        \scalebox{\commontreescale}{
        \begin{tikzpicture}[
            level 1/.style={sibling distance=2.8cm},
            level distance=2.7cm,
            edge from parent/.style={
                draw,
                line width=0.9pt,
                rounded corners=6pt,
                -{Stealth[length=3.2mm,width=2.2mm]},
            },
            snode/.style={
                rectangle,
                rounded corners=10pt,
                draw=black,
                line width=0.9pt,
                align=center,
                font=\small,
                fill=white,
                inner sep=6pt,
                blur shadow={shadow blur steps=5,shadow blur radius=3.5pt,shadow xshift=1.2pt,shadow yshift=-1.2pt}
            },
        ]
            \node[snode] (root) {Set $M$ \\ $\{g_1, \ldots, g_8\}$ \\ $u(M)=2$}
                child { node[snode] {Set $S_1$ \\ $\{g_1, g_2, g_3\}$ \\ $u(S_1)=1$} }
                child { node[snode] {Set $S_2$ \\ $\{g_4, g_5, g_6\}$ \\ $u(S_2)=1$} }
                child { node[snode] {Set $S_3$ \\ $\{g_7, g_8\}$ \\ $u(S_3)=1$} };
        \end{tikzpicture}
        }
        \caption{Size 8 Key Matroid}
        \label{fig:key_matroid_8}
    \end{minipage}
    \hfill
    \begin{minipage}{0.48\textwidth}
        \centering
        \scalebox{\commontreescale}{
        \begin{tikzpicture}[
            level 1/.style={sibling distance=2.8cm},
            level distance=2.7cm,
            edge from parent/.style={
                draw,
                line width=0.9pt,
                rounded corners=6pt,
                -{Stealth[length=3.2mm,width=2.2mm]},
            },
            snode/.style={
                rectangle,
                rounded corners=10pt,
                draw=black,
                line width=0.9pt,
                align=center,
                font=\small,
                fill=white,
                inner sep=6pt,
                blur shadow={shadow blur steps=5,shadow blur radius=3.5pt,shadow xshift=1.2pt,shadow yshift=-1.2pt}
            },
        ]
            \node[snode] (root) {Set $M$ \\ $\{g_1, \ldots, g_{12}\}$ \\ $u(M)=3$}
                child { node[snode] {Set $S_1$ \\ $\{g_1, \ldots, g_3\}$ \\ $u(S_1)=1$} }
                child { node[snode] {Set $S_2$ \\ $\{g_4, \ldots, g_6\}$ \\ $u(S_2)=1$} }
                child { node[snode] {Set $S_3$ \\ $\{g_7, \ldots, g_9\}$ \\ $u(S_3)=1$} }
                child { node[snode] {Set $S_4$ \\ $\{g_{10}, \ldots, g_{12}\}$ \\ $u(S_4)=1$} };
        \end{tikzpicture}
        }
        \caption{Size 12 Key Matroid}
        \label{fig:key_matroid_12}
    \end{minipage}
    \caption{Visualizations of the Size 8 and Size 12 Key Matroids.}
    \label{fig:key_matroids_appendix}
\end{figure}

Using these two key matroids, we will prove a result analogous to the $3$-colorable matroid decomposition.

\begin{lemma}\label{lem:4-agent-reduction}
    Let $(M,\calL,u)$ be a $4$-packed laminar matroid. If none of the splitting, child pruning, parent pruning, or partition extraction operations can be applied to $(M,\calL,u)$, then every connected component of $(M,\calL,u)$ is either a uniform matroid or one of the two key matroids.
\end{lemma}
\begin{proof}
Fix a connected component $(M',\calL',u')$ and, for the sake of contradiction, assume that it is neither a uniform matroid nor one of the two key matroids; we derive a contradiction. Since the matroid is $4$-packed, $u(S)=\ceil{|S|/4}$ for all $S \in \calL'$.

\medskip
First, we will show that every leaf of $(M',\calL',u')$ has size at most $3$.
Let $S$ be a leaf, i.e., $D(S)=\emptyset$.
If $|S|=4$, then \emph{partition extraction} applies.
If $|S|>4$, then \emph{splitting} applies to $S$ because $S$ is a leaf, so the conditions of splitting on $C(S)$ are vacuous, and thus we can let $S_1$ be any subset of $S$ of size $4$.
Hence, every leaf satisfies $|S|\le 3$.

\medskip
Select some $S \in \calL'$ such that every $C \in C(S)$ is a leaf (such a set must exist, or else $(M',\calL',u')$ would be a uniform matroid).
$C(S)$ is a disjoint partition of $S$, and each $C \in C(S)$ has $|C| \leq 3$.
Also, note that $|S|>4$, otherwise \emph{child pruning} applies.

\medskip
Next, we show that every $C \in C(S)$ satisfies $|C| > 1$.
Suppose, for contradiction, that $C(S)$ contains a size-$1$ child. Any collection of children whose sizes sum to $4$ has a union that is a proper subset of $S$, because $|S|>4$, and therefore permits \emph{splitting}. Hence, there can be no size-$3$ child and at most one size-$2$ child. If a size-$2$ child exists, there can be at most one size-$1$ child, since a size-$2$ child together with two size-$1$ children has total size $4$. If no size-$2$ child exists, there can be at most three size-$1$ children, since four such children have total size $4$. Thus, in either case, $\sum_{C \in C(S)}|C| \leq 3$, contradicting the facts that $C(S)$ partitions $S$ and $|S|>4$. Therefore, $C(S)$ contains no size-$1$ child. The same size-$4$ argument shows that $C(S)$ contains at most one size-$2$ child, since the union of two such children would permit splitting.

With this information, we will examine exactly what forms the multiset of sizes of the sets in $C(S)$ can take.

First, note that $|C(S)| > 2$. Because $(M', \calL', u')$ is packed, we have $u'(S) \geq \ceil{5/4} = 2$, and for each $C \in C(S)$, we have $u'(C) \leq \ceil{3/4} = 1$. Thus, if $|C(S)| \leq 2$, we would have $u'(S) \geq 2 \geq \sum_{C \in C(S)}u'(C)$, and the parent pruning operation would apply.

Since $|C(S)| > 2$, and we know that there can be no sets of size $1$ and at most a single set of size $2$, there must be at least two sets of size $3$.

Let $k_3 \geq 2$ be the number of size-$3$ sets in $C(S)$ and $k_2\in\{0,1\}$ indicate whether there is a size-$2$ set in $C(S)$.
Then $|S|=3k_3+2k_2$.

\smallskip
\emph{Case A: $k_2=0$ (all children have size $3$).}
If $k_2 = 0$, then by the fact that $|C(S)| \geq 3$, we must have $k_3 \geq 3$.
If $k_3=3$, then $|S|=9$ and $u'(S)=\ceil{9/4}=3$, while $\sum_{C \in C(S)} u'(C) = 3$, so \emph{parent pruning} would apply.
If $k_3\ge 5$, then the union of any four children has size $12$ (a positive multiple of $4$) and is a proper subset of $S$,
so \emph{splitting} applies. Hence the only possibility is $k_3=4$, i.e., $|S|=12$ with four size-$3$ children.

\smallskip
\emph{Case B: $k_2=1$ (one size-$2$ child and the rest size $3$).}
If $k_3\ge 3$, then two size-$3$ children plus the size-$2$ child have union size $8$ (a positive multiple of $4$),
which is a proper subset of $S$, so \emph{splitting} applies. Hence the only possibility is $k_3=2$, i.e., $|S|=8$ with child sizes $(3,3,2)$.

Thus $S$ is either (i) size $12$ with four children of size $3$, or (ii) size $8$ with children of sizes $(3,3,2)$.

If $S$ is top-level in the component, then $S$ and its children form one of the two key matroids, contradicting our assumption. Otherwise, $S$ has an immediate parent $T$, so $S \in C(T)$. Since $|S| \in \set{8,12}$ is a positive multiple of $4$, we may apply splitting to $T$ with $S_1 = S$: the child $S$ is contained in $S_1$, while every other child of $T$ is disjoint from $S_1$. This is a contradiction. Thus, we reach a contradiction in all cases.
\end{proof}

Using \Cref{lem:4-agent-reduction}, we can decompose $4$-colorable laminar matroids in the same way that we decomposed $3$-colorable ones. If we could find ways to ``solve'' the two $4$-colorable key matroids (i.e., find an EF1 allocation of the goods to $4$ agents that also has the ``no ordered envy'' property), then we could combine these results with \Cref{alg:EF1-alg} to prove the desired guarantee for $4$ agents under all $4$-colorable laminar matroids. In the lemma below, we show that this is not as straightforward as in the $3$-agent case, as the technique we use to ``solve'' the $3$-colorable key matroid no longer works.

\begin{lemma}
    It is not always possible to construct a complete and feasible allocation of the goods in the size-$8$ laminar matroid to $4$ agents such that for each agent $i \in N$, that agent's top $i$ favorite goods all appear in different bundles.
\end{lemma}
\begin{proof}
    By the constraints imposed by the size-$8$ key matroid, we must have that:
    \begin{itemize}
        \item $|A_i| = 2$ for all $i \in N$.
        \item For any set $S \in \{\{g_1,g_2,g_3\}, \{g_4,g_5,g_6\}, \{g_7,g_8\}\}$, no more than a single good from $S$ can appear in any bundle.
    \end{itemize}

    Assume that agent $4$'s top $4$ favorite goods are $\{g_1,g_2,g_5,g_6\}$ and agent $3$'s top $3$ favorite goods are $\{g_3,g_5,g_6\}$ (the top goods of agents $1$ and $2$ are arbitrary).

    From the constraint induced by agent $4$'s top $4$ favorite goods, we know that in any valid allocation, each bundle contains exactly one of the goods in $\{g_1,g_2,g_5,g_6\}$. From the constraints of the key matroid, $g_3$ cannot appear in the same bundle as $g_1$ or $g_2$. Further, from the constraint induced by agent $3$'s top $3$ favorite goods, $g_3$ cannot appear in the same bundle as $g_5$ or $g_6$. However, this means that $g_3$ cannot appear in any bundle.
\end{proof}

This shows that we cannot use the same proof technique we used in the $3$-colorable case in order to solve these larger key matroids, but it still leaves open the question of whether an allocation with the desired properties always exists.

\end{document}